\documentclass{lmcs}
\pdfoutput=1

\usepackage{lastpage}
\lmcsdoi{16}{4}{14}
\lmcsheading{}{\pageref{LastPage}}{}{}%
{Jan.~26,~2018}{Dec.~14,~2020}{}

\usepackage{proof} 
\usepackage[utf8]{inputenc}
\usepackage{stmaryrd} 
\usepackage[notext,not1]{stix}
\usepackage{listings}
\newcommand{\bff}{{\sc bff}}
\newcommand{\ocase}{{\tt case }}
\newcommand{\oletrec}{{\tt letRec }}
\newcommand{\M}{{\tt M}}
\newcommand{\N}{{\tt N}}
\newcommand{\D}{{\tt D}} 
\newcommand{\Nat}{{\tt Nat}}
\newcommand{\bN}{\mathbb{N}}
\newcommand{\cN}{\overline{\mathbb{N}}}
\newcommand{\B}{\textbf{B}}
\newcommand{\conone}{{\tt c}}
\newcommand{\comp}{{\tt comp}}
\newcommand{\icmp}{{\tt comi}}
\newcommand{\pcmp}{{\tt comp}}
\newcommand{\ecmp}{{\tt come}}

\newcommand{\letrec}[2]{{\tt letRec }\ #1 = #2}
\newcommand{\funone}{{\tt f}}
\newcommand{\varone}{{\tt x}}
\newcommand{\x}{\varone}
\newcommand{\y}{{\tt y}}
\newcommand{\z}{{\tt z}}
\newcommand{\vone}{{\tt v}}
\newcommand{\nil}{{\tt nil}}
\newcommand{\op}{{\tt op}}
\newcommand{\ckbd}{{\tt chkbd}}
\newcommand{\cut}{{\tt cut}}

\newcommand{\lcase}[3]{{\tt case }\ #1 \ {\tt  of }\  {#2}_1(\vet{\varone_1}) \to  {#3}_1, ..., {#2}_n(\vet{\varone_n}) \to {#3}_n}
\newcommand{\ccase}[1]{{\tt case }\ #1 \ {\tt  of }}
\newcommand{\some}[3]{#1_{#2}, \cdots, #1_{#3}}
\newcommand{\many}[2]{\some{#1}{1}{#2}}
\newcommand{\Var}{\mathcal{X}}
\newcommand{\vet}[1]{\overrightarrow{#1}}
\newcommand{\Cns}{\mathcal{C}}
\newcommand{\Constructors}{\Cns}
\newcommand{\suc}{ {\tt +1}}

\newcommand{\T}{{\tt T}}
\newcommand{\ord}{{\tt ord}}
\newcommand{\funtwo}{\texttt{g}}
\newcommand{\interp}[1]{\llparenthesis #1 \rrparenthesis}
\newcommand{\interpone}[1]{\llparenthesis #1 \rrparenthesis_{\rho}}
\newcommand{\transform}[2]{\lbrbrak #1 \rbrbrak_{#2}}
\newcommand{\interpcontext}[2]{\llparenthesis #1 \rrparenthesis_{#2}}
\newcommand{\taille}[1]{|#1|}
\newcommand{\FP}{\mbox{\sc fp}}

\newcommand{\sem}[1]{\llbracket #1 \rrbracket}
\newcommand{\itlp}{\text IBTLP}
\newcommand{\FPtime}{{\sc fptime }}

\begin{document}
\title[Theory of higher order interpretations and application to BFF]{Theory of higher order interpretations and application to Basic Feasible Functions}

\author[E.~Hainry]{Emmanuel Hainry}	
\address{Université de Lorraine, CNRS, Inria, LORIA, F-54000 Nancy, France}	
\email{\{emmanuel.hainry,romain.pechoux\}@loria.fr}  
\author[R.~Péchoux]{Romain Péchoux}	
\thanks{This work has been partially supported by ANR Project ELICA ANR-14-CE25-0005 and Inria associate team TC(Pro)$^3$.}

\begin{abstract} 
Interpretation methods and their restrictions to polynomials have been deeply used to control the termination and complexity of first-order term rewrite systems. This paper extends interpretation methods to a pure higher order functional language. We develop a theory of higher order functions that is well-suited for the complexity analysis of this programming language. The interpretation domain is a complete lattice and, consequently, we express program interpretation in terms of a least fixpoint. As an application, by bounding interpretations by higher order polynomials, we characterize Basic Feasible Functions at any order.
\end{abstract} 

\maketitle

\section{Introduction}
\subsection{Higher order interpretations}
This paper introduces a theory of higher order interpretations for studying higher order complexity classes. These interpretations are an extension of usual (polynomial) interpretation methods introduced in~\cite{MN70,L79} and used to show the termination of (first order) term rewrite systems~\cite{CMPU05,CL92} or to study their complexity~\cite{BMM11}. 

This theory is a novel and uniform extension to higher order functional programs: the definition works at any order on a simple programming language, where interpretations can be elegantly expressed in terms of a least fixpoint, and no extra constraints are required.

 The language has only one semantics restriction: its reduction strategy is enforced to be leftmost outermost as interpretations are non decreasing functions.  Similarly to first order interpretations, higher order interpretations ensure that each reduction step corresponds to a strict decrease. Consequently, some of the system properties could be lost if a reduction occurs under a context. 
 
\subsection{Application to higher order polynomial time}
As Church-Turing's thesis does not hold at higher order, distinct and mostly pairwise incomparable complexity classes are candidates as a natural equivalent of the notion of polynomial time computation for higher order.

The class of polynomially computable real functions by Ko~\cite{Ko91} and  the class of Basic Feasible Functional at order $i$ (\bff$_i$) by Irwin, Kapron and Royer~\cite{IKR02} belong to the most popular definitions for such classes. In~\cite{Ko91} polynomially computable real functions are defined in terms of first order functions over real numbers. They consist in order 2 functions over natural numbers and an extension at any order is proposed by Kawamura and Cook in~\cite{KC10stoc}. The main distinctions between these two models are the following:
\begin{itemize}
\item the book~\cite{Ko91} deals with representation of real numbers as input while the paper~\cite{IKR02} deals with general functions as input,
\item the book~\cite{Ko91} deals with the number of steps needed to produce an output at a given precision while the paper~\cite{IKR02} deals with the number of reduction steps needed to evaluate the program.
\end{itemize}

Moreover, it was shown in~\cite{IKR02} and~\cite{Feree14} that the classes \bff$_i$ cannot capture some functions that could be naturally considered to be polynomial time computable because they do not take into account the size of their higher order arguments. However they have been demonstrated to be robust, they characterize exactly the well-known class \FPtime of polynomial time computable functions as well as the well-known class of Basic Feasible Functions \bff, that corresponds to  order 2 polynomial time computations, and have already been characterized in various ways, \textsl{e.g.}~\cite{CK89}. 

The current paper provides a characterization of the \bff$_i$ classes as they deal with discrete data as input and they are consequently more suited to be studied with respect to usual functional languages. This result was expected to hold as it is known for a long time that (first order) polynomial interpretations characterize \FPtime and as it is shown in~\cite{FHHP15} that (first order) polynomial interpretations on stream programs characterize \bff.

\subsection{Related works}
The present paper is an extended version of the results in~\cite{HP17}: more proofs and examples have been provided. An erratum has been provided: the interpretation of the case construct has been slightly modified so that we can consider non decreasing functions (and not only strictly increasing functions).

There are two lines of work that are related to our approach. In~\cite{VDP93}, Van de Pol introduced higher order interpretation for showing the termination of higher order term rewrite systems. In~\cite{BaillotL12,BL16}, Baillot and Dal Lago introduce higher order interpretations for complexity analysis of term rewrite systems. While the first work only deals with termination properties, the second work is restricted to a family of higher order term rewrite systems called simply typed term rewrite systems. Our work can be viewed as an extension of~\cite{BL16} to functional programs and polynomial complexity at any order.

\subsection{Outline}
In Section 2, the syntax and semantics of the functional language are introduced. The new notion of higher order interpretation and its properties are described in Section 3. Next, in Section 4, we briefly recall the \bff$_i$ classes and their main characterizations, including a characterization based on the BTLP programming language of~\cite{IKR02}. Section 5 is devoted to the characterization of these classes using higher order polynomials. The soundness relies on interpretation properties: the reduction length is bounded by the interpretation of the initial term. The completeness is demonstrated by simulating a  BTLP procedure: compiling procedures to terms after applying some program transformations. 
In Section 6, we briefly discuss the open issues and future works related to higher order interpretations.

\section{Functional language}
\subsection{Syntax.}
The considered programming language consists in an unpure lambda calculus with constructors, primitive operators, a $\ocase$ construct for pattern matching and a $\oletrec$ instruction for  function definitions that can be recursive. It is as an extension of PCF~\cite{Mit96} to inductive data types and it enjoys the same properties (confluence and completeness with respect to partial recursive functions for example). 

The set of terms $\mathcal{T}$ of the language is generated by the following grammar:
$$ \M ,\N::= \x \  | \  \conone  \  | \  \op \  | \ \lcase{\M}{\conone}{\M } \  | \ \M\ \N \   | \  \lambda \x.\M \  | \  \letrec{\funone}{\M} ,$$

where $\conone,\many{\conone}{n}$ are constructor symbols of fixed arity and $\op$ is an operator of fixed arity. Given a constructor or operator symbol $b$, we write $ar(b)=n$ whenever $b$ is of arity $n$. 
$\x,\funone$ are variables in $\Var$ and $\vet{\varone_i}$ is a sequence of $ar(\conone_i)$ variables.

The free variables $FV(\M)$ of a term $\M$ are defined as usual.
Bounded variables are assumed to have distinct names in order to avoid name clashes. A closed term is a term $\M$ with no free variables, $FV(\M)=\emptyset$.\\
A substitution $\{\N_1/\varone_1,\cdots,\N_n/\varone_n \}$ is a partial function mapping variables $\many{\varone}{n}$ to terms $\many{\N}{n}$. The result of applying the substitution $\{\N_1/\varone_1,\cdots,\N_n/\varone_n \}$ to a term $\M$ is noted $\M \{\N_1/\varone_1,\cdots,\N_n/\varone_n \}$ or $\M \{\vet{\N}/\vet{\varone}\}$ when the substituting terms are clear from the context.\\

\subsection{Semantics.}\label{ss:sem}
Each  primitive operator $\op$ has a corresponding semantics $\sem{\op}$ fixed by the language implementation. $\sem{\op}$ is a total function from $\mathcal{T}^{ar(\op)}$ to $\mathcal{T}$.\footnote{Operators are total functions over terms and are not only defined on \emph{``values''}, \textsl{i.e.} terms of the shape $\lambda \x.\M$, so that we never need to reduce the operands in rule $\to_\op$. This will allow us to consider non decreasing operator interpretations in Definition~\ref{def:inter} instead of strictly increasing operator interpretation. }

We define the following relations between two terms of the language:
\begin{itemize}
\item $\beta$-reduction: $\lambda \x. \M\ \N \to_\beta \M\{\N/\x\},$
\item pattern matching: $\texttt{case}\ {\conone_j(\vet{\N_j})} \ \texttt{of} \ \ldots  \conone_j(\vet{\x_j}) \to \M_j   \ldots \to_\ocase \M_j \{\vet{\N_j}/\vet{\varone_j}\}, $
\item operator evaluation: $\op\ \M_1 \ldots \M_n \to_\op \sem{\op}(\M_1,\ldots,\M_n),$
\item fixpoint evaluation: $\letrec{\funone}{\M} \to_\oletrec \M\{ \letrec{\funone}{\M}/\funone\}.$
\end{itemize}

Let $\to_\alpha$ be defined as $\cup_{r \in \{\beta , \ocase ,\oletrec , \op \}} \to_r$. 
Let $\Rightarrow^k$ be the leftmost outermost (normal-order) evaluation strategy defined with respect to $\to_{\alpha}$ in Figure~\ref{fig:lmom}. The index $k$ accounts for the number of $\to_\alpha$ steps fired during a reduction. Let $\Rightarrow$ be a shorthand notation for $\Rightarrow^1$.
\begin{figure*}[!t]
\centering
\hrule
\vspace*{0.3cm}
\begin{tabular}{c}
$ 

\infer[]{\texttt{case}\ {\M} \ \texttt{of} \ \ldots \Rightarrow^k \texttt{case}\ {\N} \ \texttt{of} \ \ldots}{\M \Rightarrow^k \N}
\qquad
\infer[]{\M\ \N \Rightarrow^k \M'\ \N}{\M \Rightarrow^k \M'}
\qquad
\infer[]{\M \Rightarrow^1 \N}{\M \to_\alpha \N}
\qquad
\infer[]{\M \Rightarrow^{k+k'} \N}{\M\  \Rightarrow^k \M'\ \quad \M'\  \Rightarrow^{k'} \N}
$
\\[0.3cm]
\end{tabular}
\hrule
\caption{Evaluation strategy}
\label{fig:lmom}
\end{figure*}

Let $\taille{\M \Rightarrow^k \N}$ be the number of reductions distinct from $\to_\op$ in a given a derivation $\M \Rightarrow^k \N$. $\taille{\M \Rightarrow^k \N} \leq k$ always holds. 
$\sem{\M}$ is a notation for the term computed by $\M$ (if it exists), \textsl{i.e.} $\exists k, \M \Rightarrow^k \sem{\M}$ and $\not\exists \N, \sem{\M} \Rightarrow \N$. 

A (first order) value $\vone$ is defined inductively by either $\vone= \conone$, if $ar(\conone)=0$, or $\vone= \conone\ \vet{\vone}$, for $ar(\conone)>0$ values $\vet{\vone}$, otherwise.

\subsection{Type system.}

Let $\B$ be a set of basic inductive types $\texttt{b}$ described by their constructor symbol sets $\Constructors_\texttt{b}$. 
The set of simple types is defined by: $$\T ::= \texttt{b} \  |\ \T \longrightarrow \T,  \quad \text{with } \texttt{b} \in \B.$$ As usual $\longrightarrow$ associates to the right.

\begin{exa}
The type of unary numbers $\Nat$ can be described by $\Constructors_\Nat=\{0,\suc\}$, $0$ being a constructor symbol of $0$-arity and $\suc$ being a constructor symbol of $1$-arity. 

For any type $T$, $[T]$ is the base type for lists of elements of type $\T$ and has constructor symbol set $\mathcal{C}_{[\T]}=\{\nil,\conone\}$, $\nil$ being a constructor symbol of $0$-arity and $\conone$ being a constructor symbol of $2$-arity. 
\end{exa}

The type system is described in Figure~\ref{type} and proves judgments of the shape $\Gamma ; \Delta \vdash \M :: \T$ meaning that the term $\M$ has type $\T$ under the variable and constructor symbol contexts $\Gamma$ and $ \Delta$ respectively ; a variable (a constructor, respectively) context being a partial function that assigns types to variables (constructors and operators, respectively). 

As usual, the input type and output type of constructors and operators of arity $n$ will be restricted to basic types. Consequently, their types are of the shape $\texttt{b}_1 \longrightarrow \ldots  \longrightarrow \texttt{b}_n  \longrightarrow \texttt{b}$. A well-typed term will consist in a term $\M$ such that $\emptyset ; \Delta \vdash \M :: \T$. Consequently, it is mandatory for a term to be closed in order to be well-typed.

In what follows, we will consider only well-typed terms.
The type system assigns types to all the syntactic constructions of the language and ensures that a program does not go wrong.
Notice that the typing discipline does not prevent a program from diverging.

\begin{figure*}[!t]
\centering
\hrule
\vspace*{0.3cm}
\begin{tabular}{c}
$ 
\infer[(\texttt{Var})]{\Gamma;\Delta\vdash \varone::\T}{\Gamma(\varone)=\T}
\qquad
\infer[(\texttt{Cons})]{\Gamma;\Delta\vdash \conone::\T}{\Delta(\conone)=\T}
\qquad
\infer[(\texttt{Op})]{\Gamma;\Delta\vdash \op::\T}{\Delta(\op)=\T}
$
\\
\\
$
\infer[(\texttt{App})]{\Gamma;\Delta\vdash \M\ \N:: \T_2 
}{\Gamma;\Delta\vdash \M::\T_1 \longrightarrow \T_2 \quad \Gamma;\Delta\vdash \N::\T_1
}$
\\
\\
$
\infer[(\texttt{Abs})]{\Gamma;\Delta\vdash \lambda \x.\M::\T_1 \to \T_2 
}{\Gamma, \x::\T_1;\Delta\vdash \M:: \T_2 
}
\qquad
 \infer[(\texttt{Let})]{\Gamma;\Delta\vdash \letrec{\funone}{\M}\ ::\T }{
\Gamma , \funone:: \T ;\Delta \vdash \M::\T
}
$
\\
\\
$
\infer[(\texttt{Case})]{\Gamma;\Delta\vdash\lcase{\M}{\conone}{\M}::\T}{
\Gamma;\Delta\vdash\M::\texttt{b} & &\Gamma;\Delta \vdash\conone_i::\vet{\texttt{b}_i} \longrightarrow \texttt{b} &\Gamma, \vet{\varone_i}::\vet{\texttt{b}_i};\Delta\vdash\M_i:: \T \ (1\leq i\leq m) }
$
\\[0.3cm]
\end{tabular}
\hrule
\caption{Type system}
\label{type}
\end{figure*}

\begin{defi}[Order]\label{def:order}
The order of a type $\T$, noted $\ord(\T)$, is defined inductively by:
\begin{align*}
\ord(\texttt{b}) &=0, &&\text{if }\texttt{b} \in \B , \\
\ord(\T \longrightarrow \T')&= \max(\ord(\T)+1,\ord(\T')),&&\text{otherwise.}
\end{align*}
Given a term $\M$ of type $\T$, \textsl{i.e.} $\emptyset ; \Delta \vdash \M :: \T$, the order of $\M$ with respect to $\T$ is $\ord(\T)$.
\end{defi}

\begin{exa}\label{ex1}
Consider the following term $\M$ that maps a function to a list given as inputs:
\begin{align*}
\letrec{\funone}{\lambda \funtwo.\lambda \x.\ccase{\x}\  {}\conone(\y,\z) &\to \conone \ (\funtwo\  \y)\ (\funone\ \funtwo\ \z) },\\
\nil &\to \nil .
\end{align*}
Let $[\Nat]$ is the base type for lists of natural numbers of constructor symbol set $\mathcal{C}_{[\Nat]}=\{\nil,\conone\}$. 
The term ${\M}$ can be typed by $\emptyset ; \Delta \vdash \M :: (\Nat \longrightarrow \Nat) \longrightarrow [\Nat] \longrightarrow [\Nat]$, as illustrated by the following typing derivation:

$$
\infer[(\texttt{Let})]{\emptyset;\Delta\vdash {\M}\ ::(\Nat \longrightarrow \Nat) \longrightarrow [\Nat] \longrightarrow [\Nat]}{
	\infer[(\texttt{Abs})]{\funone::\T;\Delta\vdash \lambda \funtwo.\lambda \x.\ccase{\x}\  {}\ \conone(\y,\z) \to \conone\ (\funtwo\  \y)\ (\funone\ \funtwo\ \z), \nil \to \nil  \ ::\T }{
		\infer[(\texttt{Abs})]{\Gamma;\Delta\vdash \lambda \x.\ccase{\x}\  {}\ \conone(\y,\z) \to \conone\ (\funtwo\  \y)\ (\funone\ \funtwo\ \z), \nil \to \nil \ :: [\Nat] \longrightarrow [\Nat] }{ 
			\infer[(\texttt{Case})]{\Gamma,\x::[\Nat];\Delta\vdash  \ccase{\x}\  {}\ \conone(\y,\z) \to \conone\ (\funtwo\  \y)\ (\funone\ \funtwo\ \z), \nil \to \nil :: [\Nat]}{
							\infer[(\texttt{App})]{\cdots \qquad \Gamma,\y::\Nat,\z::[\Nat];\Delta\vdash  \conone\ (\funtwo\  \y)\ (\funone\ \funtwo\ \z) :: [\Nat]}{
							\cdots \qquad
							\infer[(\texttt{App})]{\Gamma';\Delta\vdash  \conone\ (\funtwo\  \y) :: [\Nat] \longrightarrow [\Nat]}{
							 	\infer[(\texttt{App})]{\cdots \qquad \Gamma';\Delta\vdash  (\funtwo\  \y) :: \Nat}{\cdots}
							}
							\qquad
							\infer[(\texttt{App})]{\Gamma';\Delta\vdash  (\funone\ \funtwo\ \z) :: [\Nat] }{
			\cdots
							}
						}
			}
		}
	}
}
$$

where $\T$ is a shorthand notation for $(\Nat \longrightarrow \Nat) \longrightarrow [\Nat] \longrightarrow [\Nat]$, where the derivation of the base case $\nil$ has been omitted for readability and where the contexts $\Delta,\Gamma,\Gamma'$ are such that $\Delta(\conone)= \Nat \longrightarrow [\Nat] \longrightarrow [\Nat]$, $\Delta(\nil)=[\Nat]$ and $\Gamma=\funone::\T,\funtwo::\Nat \longrightarrow \Nat$ and $\Gamma'=\Gamma,\y::\Nat,\z::[\Nat]$.  Consequently, the order of $\M$ (with respect to $\T$) is equal to $2$, as $\ord(\T)=2$.
\end{exa}

\section{Interpretations}
\subsection{Interpretations of types}
We briefly recall some basic definitions that are very close from the notions used in denotational semantics (See~\cite{W93}) since, as we shall see later, interpretations can be defined in terms of fixpoints. Let $(\bN, \leq,\sqcup,\sqcap)$ be the set of natural numbers equipped with the usual ordering $\leq$, a max operator $\sqcup$ and min operator $\sqcap$ and let $\cN$ be $\bN \cup \{ \top\}$, where $\top$ is the greatest element satisfying $\forall n \in \cN,\ n \leq \top$, $n \sqcup \top= \top \sqcup n = \top$ and $n \sqcap \top = \top \sqcap n =n$.

The \emph{interpretation} of a type  is defined inductively by:
\begin{align*}
\interp{\texttt{b}}&=\cN, &&\text{if }\texttt{b}\text{ is a basic type,}\\
 \interp{\T \longrightarrow \T'}&=\interp{\T} \longrightarrow^{\uparrow} \interp{\T'}, &&\text{otherwise,}
\end{align*}
where $\interp{\T} \longrightarrow^{\uparrow} \interp{\T'}$ denotes the set of total non decreasing functions from $\interp{\T} $ to $\interp{\T'}$. A function $F$ from the set $A$ to the set $B$ being non decreasing if for each $X,Y \in A,$ $X \leq_{A} Y$ implies $F(X) \leq_{B} F(Y)$, where $\leq_{A}$ is the usual pointwise ordering induced by $\leq$ and defined by:
\begin{align*}
n \leq_{\cN} m&\text{ iff } n \leq m,\\
F \leq_{A \longrightarrow^{\uparrow} B} G &\text{ iff } \forall X \in A,\ F(X) \leq_{B} G(X).
\end{align*}

\begin{exa}
The type $\T=(\Nat \longrightarrow \Nat) \longrightarrow [\Nat] \longrightarrow [\Nat]$ of the term $\letrec{\funone}{\M}$ in Example~\ref{ex1} is interpreted by:
$$
\interp{\T}= (\cN  \longrightarrow^{\uparrow} \cN)  \longrightarrow^{\uparrow} (\cN  \longrightarrow^{\uparrow} \cN).$$
\end{exa}

In what follows, given a sequence $\vet{F}$ of $m$ terms in the interpretation domain  and a sequence $\vet{\T}$ of $k$ types, the notation $\vet{F} \in \vet{\interp{\T}}$ means that both $k=m$ and $\forall i \in [1,m],\ F_i \in \interp{\T_i}$.

\subsection{Interpretations of terms}

\newcommand{\fleche}{\longrightarrow}
\newcommand{\Tau}{\T}
Each closed term of type $\T$ will be interpreted by a function in ${\interp{\T}}$. The application is denoted as usual whereas we use the notation $\Lambda$ for abstraction on this function space in order to avoid confusion between terms of our calculus and objects of the interpretation domain. Variables of the interpretation domain will be denoted using upper case letters. When needed, Church typing discipline will be used in order to highlight the type of the bound variable in a lambda abstraction.

An important distinction between the terms of the language and the objects of the interpretation domain lies in the fact that beta-reduction is considered as an equivalence relation on (closed terms of) the interpretation domain, \textsl{i.e.} $(\Lambda X. F)\ G = F\{G/X\}$ underlying that $(\Lambda X. F)\ G$ and $ F\{G/X\}$  are distinct notations that represent the same higher order function. The same property holds for $\eta$-reduction, \textsl{i.e.} $ \Lambda X. (F\ X)$ and $F$ denote the same function.

In order to obtain complete lattices,  each type ${\interp{\T}}$ has to be completed by a lower bound $\bot_{\interp{\T}}$ and an upper bound $\top_{\interp{\T}}$ as follows:
\begin{align*}
\bot_{\cN}&=0, \\
 \top_{\cN}&=\top,\\
\bot_{\interp{\T \longrightarrow \T'}}&= \Lambda X^{\interp{\T}}.\bot_{\interp{\T'}}, \\ 
\top_{\interp{\T \longrightarrow \T'}}&= \Lambda X^{\interp{\T}}.\top_{\interp{\T'}}.
\end{align*}

\begin{lem}
For each $\T$ and for each $F \in \interp{\T}$, $\bot_{\interp{\T}} \leq_{\interp{\T}} F \leq_{\interp{\T}} \top_{\interp{\T}}$.
\end{lem}
\begin{proof}
By induction on types.
\end{proof}
 Notice that for each type $\T$ it also holds that $\top_{\interp{\T}} \leq_{\interp{\T}} \top_{\interp{\T}}$, by an easy induction.

In the same spirit, $\max$ and $\min$ operators $\sqcup$ (and $\sqcap$) over $\cN$ can be extended to higher order functions $F,G$ of any arbitrary type $\interp{ \T} \longrightarrow^{\uparrow} \interp{\T'}$ by:
\begin{align*}
\sqcup^{\interp{ \T} \longrightarrow^{\uparrow} \interp{\T'}}(F,G) &= \Lambda X^{\interp{\T}}.\sqcup^{\interp{\T'}}( F(X), G(X)),\\
\sqcap^{\interp{ \T} \longrightarrow^{\uparrow} \interp{\T'}}(F,G) &= \Lambda X^{\interp{\T}}.\sqcap^{\interp{\T'}}( F(X), G(X)).
\end{align*}
In the following, we use the notations $\bot$, $\top$, $\leq$, $<$, $\sqcup$ and $\sqcap$ instead of $\bot_{\interp{\T}}$, $\top_{\interp{\T}}$, $\leq_{\interp{\T}}$, $<_{\interp{\T}}$, $\sqcup^{\interp{\T}}$ and $\sqcap^{\interp{\T}}$, respectively, when $\interp{\T}$ is clear from the typing context. Moreover, given a boolean predicate $P$ on functions, we will use the notation $\sqcup_P\{F\}$ as a shorthand notation for $\sqcup\{ F \ | \  P\}$.

\begin{lem}\label{lattice}
For each type $\T$, $(\interp{\T},\leq,\sqcup,\sqcap,\top,\bot)$ is a complete lattice.
\end{lem}
\begin{proof}
Consider a subset $S$ of elements in $\interp{\T}$ and define $\sqcup S = \sqcup_{F \in S}F$. By definition, we have $F \leq \sqcup S$, for any $F \in S$. Now consider some $G$ such that for all $F \in S$, $F \leq G$. We have $\forall F \in S$, $\forall X,$ $F(X) \leq G(X)$. Consequently, $\forall X,\ S(X)=\sqcup_{F \in S} F(X) \leq G(X)$ and $S$ is a supremum. The same holds for the infimum.
\end{proof}

Now we need to define a unit (or constant) cost function for any interpretation of type $\T$. For that purpose, let $+$ denote natural number addition extended to $\cN$ by $\forall n,\ \top+n=n+\top=\top$.
For each type $\interp{\T}$, we define inductively a dyadic sum function $\oplus_{\interp{\T}}$ by:
\begin{align*}
 X^{\cN} \oplus_{{\cN}} Y^{\cN} &= X+Y,\\
F \oplus_{\interp{\T \longrightarrow \T'}}  G  &= \Lambda X^{\interp{\T}}. (F (X)\oplus_{\interp{\T'}}G (X)).
\end{align*}

Let us also define the constant function $n_{\interp{\T}}$, for each type $\T$ and each integer $n \geq 1$, by:
\begin{align*}
n_{\cN} &=n,\\
n_{\interp{\T \longrightarrow \T'}}  &= \Lambda X^{\interp{\T}}. n_{\interp{\T'}}.
\end{align*}
Once again, we will omit the type when it is unambiguous using the notation $n \oplus$ to denote the function $n_{\interp{\T}} \oplus_{\interp{\T}}$ when $\interp{\T}$ is clear from the typing context. 

For each type $\interp{\T}$, we can define a strict ordering $<$ by: $F <G$ whenever $1 \oplus F \leq G$.

\begin{defi}
A \emph{variable assignment}, denoted $\rho$, is a map associating to
each $\funone \in\Var$ of type $\T$ a variable $F$ of type $\interp{\T}$.
\end{defi}

Now we are ready to define the notions of variable assignment and interpretation of a term $\M$. 

\begin{defi}[Interpretation]\label{def:inter} 
    Given a variable assignment $\rho$, an \emph{interpretation} is the
    extension of $\rho$ to well-typed terms, mapping each term of type $\T$ to an object in $\interp{\T}$ and defined by:
  \begin{itemize}
     \item $ \interpcontext{ \funone}{\rho}=\rho(\funone),\  \text{ if } \funone \in \Var$,\\
    \item $\interpone{\conone}=1 \oplus (\Lambda X_1.\ldots. \Lambda X_n.\sum_{i=1}^{n}X_i ),\ \text{ if }ar(\conone)=n$, \\
\item $\interpcontext{\M \N}{\rho}=\interpone{\M} \interpone{\N}$,\\
\item $\interpone{\lambda \x. \M}= 1 \oplus (\Lambda \interpone{\varone}. \interpcontext{\M}{\rho})$,\\
     \item $\interpone{\texttt{case}\ \M \ \texttt{of} \ldots \conone_j(\vet{\x_j}) \to \M_j \ldots} =1 \oplus  \sqcup_{1\leq i\leq m}\sqcup_{\interpcontext{\M}{\rho}\geq \interpone{\conone_i}\vet{F_i}}\{\interpone{\M} \oplus \interpone{\M_i}\{\vet{F_i}/\vet{\interpone{\x_i}}\}\}$,\\
   \item $ \interpcontext{\letrec{\funone}{\M}}{\rho}= \sqcap\{ F \in {\interp{\T}}  \ | \ F \geq 1 \oplus((\Lambda \interpone{\funone}. \interpone{\M})\  F)\}$,
    \end{itemize}    
where $\interpone{{\op}}$ is a non decreasing total function such that: $$\forall \M_1,\ldots,\forall \M_n, \ \interpone{{\op}\ {\M_1}\ \ldots \ {\M_n}} \geq \interpone{\sem{\op}(\M_1,\ldots,\M_n)}.$$
\end{defi}

The aim of the interpretation of a term is to give a bound on its computation time as we will shortly see in Corollary~\ref{coro2}. For that purpose, it requires a strict decrease of the interpretation under $\Rightarrow$. This is the object of Lemma~\ref{lem2}. This is the reason why any \emph{``construct''} of the language involves a $1 \oplus$ in each rule of Definition~\ref{def:inter}. Application that plays the role of a \emph{``destructor''} does not require this. This is also the reason why the interpretation of a constructor symbol does not depend on its nature.

Remark that the condition on the the interpretation of operators correspond to the notion of \emph{sup-interpretation}  (See~\cite{P13} for more details).

\subsection{Existence of an interpretation}

The interpretation of a term is always defined. Indeed, in Definition~\ref{def:inter}, $\interpcontext{\letrec{\funone}{\M}}{\rho}$ is defined in terms of the least fixpoint of the function $\Lambda X^{\interp{\T}}. 1 \oplus_{\interp{\T}} ((\Lambda \interpone{\funone}.   \interpone{\M})\ X)$ and, consequently, we obtain the following result as a direct consequence of Knaster-Tarski~\cite{Tar55,KS01} Fixpoint Theorem:
\begin{prop}\label{existence}
Each term $\M$ of type $\T$ has an interpretation.
\end{prop}
\begin{proof}
By Lemma~\ref{lattice}, $L=(\interp{\T},\leq,\sqcup,\sqcap,\top,\bot)$ is a complete lattice. The function  $F =\Lambda X^{\interp{\T}}. 1 \oplus_{\interp{\T}} ((\Lambda \interpone{\funone}.   \interpone{\M})\ X) : L \to L$ is monotonic. Indeed, both constructor terms and {\tt letRec} terms of type $\interp{\T}$ are interpreted over a space of monotonic functions ${\interp{\T}}$. Moreover monotonicity is preserved by application, abstraction and the $\sqcap$ and $\sqcup$ operators. Applying Knaster-Tarski, we obtain that $F$ admits a least fixpoint, which is exactly $\sqcap\{ X \in {\interp{\T}}  \ | \ X \geq  F X\}$.
\end{proof}

\subsection{Properties of interpretations}
We now show intermediate lemmata.
The following Lemma can be shown by structural induction on terms:
\begin{lem}\label{lem1}
For all $\M, \N, \x$ such that $\x :: \T ; \Delta \vdash \M :: \T'$, $\emptyset ; \Delta \vdash \N :: \T$, we have: $$ \interpone{\M}\{\interpone{\N}/\interpone{\varone}\}=\interpone{\M\{{\N}/{\varone}\}}.$$
\end{lem}

\begin{lem}\label{decr}
For all $\M, \N, \x$ such that $\x :: \T ; \Delta \vdash \M :: \T'$, $\emptyset ; \Delta \vdash \N :: \T$, we have $ \interpone{\lambda \varone. \M\ \N} >\interpone{\M\{\N/\varone\}}$.
\end{lem}
\begin{proof}
\begin{align*}
\interpone{\lambda \varone. \M\ \N} &= \interpone{\lambda \varone. \M} \interpone{ \N} &&\text{(By Definition~\ref{def:inter})} \\
&=(\Lambda \interpone{\varone}.1 \oplus \interpone{\M})\ \interpone{\N} &&\text{(By Definition~\ref{def:inter}) } \\
 &=1 \oplus\interpone{\M}\{\interpone{\N}/\interpone{\varone}\} &&\text{(By definition of } =)\\
&=1 \oplus\interpone{\M\{{\N}/{\varone}\}}  &&\text{(By Lemma~\ref{lem1})}\\
&>\interpone{\M\{{\N}/{\varone}\}} &&\text{(By definition of }>)
\end{align*}
and so the conclusion. 
\end{proof}

\begin{lem}\label{lem2}
 For all $\M$, we have: if $\M \Rightarrow \N$ then $\interpone{\M} \geq \interpone{\N}$. 
Moreover if $\taille{\M \Rightarrow \N}=1$ then $\interpone{\M} > \interpone{\N}$.
\end{lem}
\begin{proof}
If $\taille{\M \Rightarrow \N}=0$ then $\M = \op\ \M_1\ \ldots\ \M_n \to_\op \sem{\op}(\M_1,\ldots,\M_n)=\N$, for some operator $\op$ and terms $\M_1,\ldots,\M_n$ and consequently, by Definition of interpretations we have: $$ \interpone{{\op}\ {\M_1}\ \ldots \ {\M_n}} \geq \interpone{\sem{\op}(\M_1,\ldots,\M_n)}.$$

If $\taille{\M \Rightarrow \N}=1$ then the reduction is not $\to_\op$.
By Lemma~\ref{decr}, in the case of a $\beta$-reduction and, by induction, by Lemma~\ref{lem1} and Definition~\ref{def:inter} for the other cases. \textsl{e.g.} For a {\tt letRec} reduction, we have: if $\M=\letrec{\funone}{\M'} \to_\oletrec {\M'}\{ \M/\funone\} = \N$ then:
\begin{align*}
\interpone{\M} &= \sqcap\{ F \in \interp{\T}  \ | \ F \geq 1 \oplus ((\Lambda \interpone{\funone}. \interpone{\N})\ F)\}\\
&\geq  \sqcap\{ 1 \oplus ((\Lambda \interpone{\funone}. \interpone{\N})\ F) \ | \ F \geq 1 \oplus ((\Lambda \interpone{\funone}. \interpone{\N})\ F)\} \\
&\geq 1 \oplus  ((\Lambda \interpone{\funone}. \interpone{\N})\  \sqcap \{ F \ | \ F \geq 1 ((\oplus \Lambda \interpone{\funone}. \interpone{\N})\ F)\})\\
&\geq 1 \oplus  ((\Lambda \interpone{\funone}. \interpone{\N})\ \interpone{\M})\\
&\geq 1 \oplus \interpone{\N}\{\interpone{\M}/\interpone{\funone}\}\\
&\geq 1 \oplus \interpone{\N\{\M/\funone\}} && \text{(By Lemma}~\ref{lem1})\\
&> \interpone{\N\{\M/\funone\}}. &&\text{(By definition of }>)
\end{align*}
The first inequality holds since we are only considering higher order functions $F$ satisfying $F \geq 1 \oplus ((\Lambda \interpone{\funone}. \interpone{\N})\  F)$.
The second inequality holds because $ \Lambda \interpone{\funone}. \interpone{\N}$ is a non decreasing function (as the interpretation domain only consists in such functions).
\end{proof}

As each reduction distinct from an operator evaluation corresponds to a strict decrease, the following corollary can be obtained:
\begin{cor}\label{coro1}
For all terms, $\M$, $\vet{\N}$, such that $ \emptyset ; \Delta \vdash \M\ \vet{\N} :: \T$, if $ \M\ \vet{\N} \Rightarrow^k \M'$ then $\interpone{\M} \interpone{\vet{\N}} \geq \taille{\M\ \vet{\N} \Rightarrow^k \M'}\oplus \interpone{\M'}.$
\end{cor}

As basic operators can be considered as constant time computable objects the following Corollary also holds:
\begin{cor}\label{coro2}
For all terms, $\M$, $\vet{\N}$, such that $ \emptyset ; \Delta \vdash \M\ \vet{\N} :: \texttt{b}$, with $\texttt{b} \in \B$, if $\interpone{\M\ \vet{\N}} \neq \top$ then $\M\ \vet{\N}$ terminates in a number of reduction steps in $O(\interpone{\M\ \vet{\N}})$.
\end{cor}

The size $\taille{\vone}$ of a value $\vone$ (introduced in Subsection~\ref{ss:sem}) is defined by $\taille{\conone}=1$ and $\taille{\M\ \N}=\taille{\M}+\taille{\N}$.

\begin{lem}\label{lem7}
For any value $\vone$, such that $ \emptyset ; \Delta \vdash \vone  :: \texttt{b}$, with $\texttt{b} \in \B$, we have
$\interp{\vone} = \taille{\vone}$.
\end{lem}

\begin{exa}\label{prevv}
Consider the following term $\M :: \Nat \longrightarrow \Nat$ computing the double of a unary number given as input:
\begin{align*}
\letrec{\funone}{ \lambda \x.\ccase{\x}\  {}\suc(\y) &\to \suc(\suc(\funone\  \y)) },\\
\ 0 &\to 0.
\end{align*}
We can see below how the interpretation rules of Definition~\ref{def:inter} are applied on such  a term. 
\begin{align*}
\interpone{\M}&=\sqcap\{ F   \ | \ F \geq 1 \oplus   ((\Lambda \interpone{\funone}.\interpone{\lambda \x.\ccase{\x}\  {}\suc(\y) \to \suc(\suc(\funone\  \y)) | \ 0 \to 0}) F)\}\\
&=\sqcap\{ F \ | \ F \geq 2 \oplus( \Lambda \interpone{\x}. \interpone{\ccase{\x}\  {}\suc(\y) \to \suc(\suc(\funone\  \y)) | \ 0 \to 0} \{F/\interpone{\funone}\})\}\\
&=\sqcap\{ F  \ | \ F \geq 3 \oplus( \Lambda X. X\oplus((\sqcup_{X\geq \interpone{\suc(\y)}}\interpone{\suc(\suc(\funone\  \y))})\sqcup (\sqcup_{X\geq \interpone{0}} \interpone{0} ) \{F/\interpone{\funone}\})\}\\
&=\sqcap\{ F  \ | \ F \geq 3 \oplus(\Lambda X.X\oplus( (\sqcup_{X\geq 1 \oplus \interpone{\y}} 2 \oplus (\interpone{\funone}\  \interpone{\y}))\sqcup (\sqcup_{X\geq 1} 1 ) \{F/\interpone{\funone}\})\}\\
&=\sqcap\{ F \ | \ F \geq 3 \oplus(\Lambda X. X\oplus((\sqcup_{X\geq 1 \oplus \interpone{\y}} 2 \oplus (F\  \interpone{\y}))\sqcup (1 ))\}\\
&=\sqcap\{ F \ | \ F \geq 3 \oplus(\Lambda X. X\oplus(( 2 \oplus (F\  (X-1)))\sqcup(1))\}, \quad X-1 \geq 0\\
&=\sqcap\{ F  \ | \ F \geq  \Lambda X. (5 \oplus X \oplus (F\  (X-1))) \sqcup (4\oplus X) \}\\
&\leq \Lambda X.6X^2+5
\end{align*}

In the end, we search for the minimal non decreasing function $F$ greater than $\Lambda X. (5 \oplus X \oplus (F\  (X-1))) \sqcup (4\oplus X) $, for $X>1$. As the function $\Lambda X.6X^2\oplus5$ is a solution of such an inequality, the fixpoint is smaller than this function. Notice that such an interpretation is not tight as one should have expected the interpretation of such a program to be $\Lambda X.2X$. This interpretation underlies that:
\begin{itemize}
\item the iteration steps, distinct from the base case, count for $5\oplus X$: 1 for the {\tt letRec} call, 1 for the application, $1\oplus X$ for pattern matching and 2 for the extra-constructors added,
\item the base case counts for $4\oplus X$: 1 for recursive call, 1 for application, $1\oplus X$ for pattern matching and $1$ for the constructor.
\end{itemize}
 Consequently, we have a bound on both size of terms and reduction length though this upper bound is not that accurate. This is not that surprising as this technique suffers from the same issues as first-order interpretation based methods.

\end{exa}

\subsection{Relaxing interpretations}
For a given program it is somewhat difficult to find an interpretation that can be expressed in an easy way. This difficulty lies in the homogeneous definition of the considered interpretations using a max (for the case construct) and a min (for the {\tt letRec} construct). Indeed, it is sometimes difficult to eliminate the constraint (parameters) of a max generated by the interpretation of a case. Moreover, it is a hard task to find the fixpoint of the interpretation of a {\tt letRec}. All this can be relaxed as follows:
\begin{itemize}
\item finding an upper-bound of the max by eliminating constraints in the case construct interpretation,
\item taking a particular function satisfying the inequality in the {\tt letRec} construct interpretation.
\end{itemize}
In both cases, we will no longer compute an exact interpretation of the term but rather an upper bound of the interpretation.

\begin{lem}
Given a set of functions $S$ and a function $F \in S$, the following inequality always holds
$F \geq \sqcap\{ G | G \in S\}.$
\end{lem}
This relaxation is highly desirable in order to find ``lighter'' upper-bounds on the interpretation of a term.
Moreover, it is a reasonable approximation as we are interested in worst case complexity.
Obviously, it is possible by relaxing too much to attain the trivial interpretation $\top_{\interp{\T}}$.
Consequently, these approximations have to be performed with moderation as taking too big intermediate upper bounds might lead to an uninteresting upper bound on the interpretation of the term.

\begin{exa}\label{prev}
Consider the term $\M$ of Example~\ref{ex1}:
\begin{align*}
\letrec{\funone}{\lambda \funtwo.\lambda \x.\ccase{\x}\  {}\conone(\y,\z) &\to \conone\ (\funtwo\  \y)\ (\funone\ \funtwo\ \z) },\\
 \nil &\to \nil.
\end{align*}
Its interpretation $\interpone{\M}$ is equal to:
\begin{align*}
&=\sqcap\{ F \ | \ F \geq 1 \oplus ((\Lambda \interpone{\funone}.  \interpone{ \lambda \funtwo.\lambda \x.\ccase{\x}\  {}\conone(\y,\z) \to \conone(\funtwo\  \y)(\funone\ \funtwo\ \z) | \ \nil \to \nil })\ F)\}\\
&=\sqcap\{ F\ | \ F \geq 3 \oplus((\Lambda \interpone{\funone}. \Lambda \interpone{\funtwo}.\Lambda \interpone{\x}. \interpone{\ccase{\x}\  {}\conone(\y,\z) \to \conone(\funtwo\  \y)(\funone\ \funtwo\ \z) | \ \nil \to \nil })\ F)\}\\
&=\sqcap\{ F  \ | \ F \geq 4 \oplus(\Lambda \interpone{\funone}. \Lambda \interpone{\funtwo}.\Lambda \interpone{\x}. \interpone{\x}\oplus (\sqcup(\interpone{\nil}, \sqcup_{ \interpone{\x} \geq \interpone{\conone(\y,\z)}}(\interpone{\conone(\funtwo\  \y)(\funone\ \funtwo\ \z) })))\ F)\}\\
&=\sqcap\{ F  \ | \ F \geq 4 \oplus(\Lambda \interpone{\funone}. \Lambda \interpone{\funtwo}.\Lambda \interpone{\x}.  \interpone{\x}\oplus (\sqcup(1, \sqcup_{ \interpone{\x} \geq 1 \oplus \interpone{\y}+\interpone{\z)}}(1\oplus (\interpone{\funtwo}\  \interpone{\y})\\
&\phantom{=}+(\interpone{\funone}\ \interpone{\funtwo}\ \interpone{\z})))\ F))\}\\
&=\sqcap\{ F   \ | \ F \geq 5 \oplus(\Lambda \interpone{\funtwo}.\Lambda \interpone{\x}.  \interpone{\x}\oplus ( \sqcup_{ \interpone{\x} \geq 1 \oplus \interpone{\y}+\interpone{\z)}}((\interpone{\funtwo}\  \interpone{\y})+(F\ \interpone{\funtwo}\ \interpone{\z}))))\}\\
&\leq \sqcap\{ F  \ | \ F \geq 5 \oplus(\Lambda \interpone{\funtwo}.\Lambda \interpone{\x}.  \interpone{\x}\oplus ( ((\interpone{\funtwo}\  (\interpone{\x}-1))+(F\ \interpone{\funtwo}\ (\interpone{\x}-1)))))\}\\
&\leq \Lambda \interpone{\funtwo}.\Lambda \interpone{\x}.(5 \oplus (\interpone{\funtwo} \ \interpone{\x}))\times (2 \times \interpone{\x})^2.
\end{align*}

In the penultimate line, we obtain an upper-bound on the interpretation by approximating the case interpretation, substituting $\interpone{x}-1$ to both $\interpone{y}$ and $\interpone{z}$. This is a first step of relaxation where we find an upper bound for the max. The below inequality holds for any non decreasing function $F$: $$\underbrace{\sqcup_{ \interpone{\x} \geq 1 \oplus \interpone{\y}+\interpone{\z)}} F(\interpone{\y},\interpone{\z})}_{a} \leq \underbrace{\sqcup_{ \interpone{\x} \geq 1 \oplus \interpone{\y}, \interpone{\x} \geq 1 \oplus \interpone{\z)}} F(\interpone{\y},\interpone{\z})}_{b}.$$ Consequently, $\sqcup \{ F \ | \ F \geq b \} \leq \sqcup \{ F \ | \ F \geq a \}$. 

In the last line, we obtain an upper-bound on the interpretation by approximating the {\tt letRec} interpretation, just checking that the function $F=\Lambda \interpone{\funtwo}.\Lambda \interpone{\x}.(5 \oplus (\interpone{\funtwo} \ \interpone{\x}))\times (2 \times \interpone{\x}) ^2$, where $\times$ is the usual multiplication symbol over natural numbers, satisfies the inequality: $$F \geq 5 \oplus(\Lambda \interpone{\funtwo}.\Lambda \interpone{\x}. \interpone{\varone}\oplus ( ((\interpone{\funtwo}\  (\interpone{\x}-1))+(F\ \interpone{\funtwo}\ (\interpone{\x}-1))))).$$

\end{exa}

\subsection{Higher Order Polynomial Interpretations}
At the present time, the interpretation of a term of type $\T$ can be any total functional over $\interp{\T}$. In the next section, we will concentrate our efforts to study polynomial time at higher order. Consequently, we need to restrict the shape of the admissible interpretations to higher order polynomials which are the higher order counterpart to polynomials in this theory of complexity.

\begin{defi}\label{def:hop}
We consider types built from the basic type $\cN$ as follows:
$$A,B  ::= \cN \ | \ A \longrightarrow B.$$
Higher order polynomials are built by the following grammar:
$$ P,Q ::= c^{\cN} \ | \ X^A \ | \ +^{\cN \longrightarrow \cN \longrightarrow \cN} \ | \ \times^{\cN \longrightarrow \cN \longrightarrow \cN}\ | \ (P^{A \longrightarrow B}\ Q^A)^B \ | \ (\Lambda X^A.P^B)^{A \longrightarrow B}.$$
where $c$ represents constants in $\bN$ and $P^A$ means that $P$ is of type $A$. 
A polynomial $P^A$ is of order $i$ if $\ord(A)=i$. When $A$ is explicit from the context, we use the notation $P_i$ to denote a polynomial of order $i$.
\end{defi}
In the above definition, constants of type $\cN$ are distinct from $\top$.
By definition, a higher order polynomial $P_{i}$ has arguments of order at most $i-1$. 
For notational convenience, we will use the application of $+$ and $\times$ with an infix notation as in the following example.
\begin{exa}
Here are several examples of polynomials generated by the grammar of Definition~\ref{def:hop}:
\begin{itemize}
\item $P_1 = \Lambda X_0.(6 \times X_0^2+5)$ is an order $1$ polynomial,
\item $Q_1 = \times$ is an order $1$ polynomial,
\item $P_2 = \Lambda X_1.\Lambda X_0.(3 \times (X_1\ (6 \times X_0^2+5))+X_0)$ is an order $2$ polynomial,
\item $Q_2 = \Lambda X_1.\Lambda X_0.((X_1\ (X_1\ 4))+(X_1\ X_0))$ is an order $2$ polynomial.
\end{itemize}
\end{exa}

We are now ready to define the class of functions computed by terms admitting an interpretation that is (higher order) polynomially bounded:
\begin{defi}
Let $\FP_i$, $i\geq0$, be the class of polynomial functions at order $i$ that consist in functions computed by a term $\emptyset;\Delta \vdash \M :: \T$ over the basic type ${\Nat}$ and such that:
\begin{itemize}
\item $\ord(\T )=i$,
\item $\interpone{\M}$ is bounded by an order $i$ polynomial (\textsl{i.e.} $\exists P_i,\ \interpone{\M} \leq P_i$).
\end{itemize}
\end{defi}

\begin{exa}
The term $\M$ of Example~\ref{prev} has order $1$ and admits an interpretation bounded by $P_1=\Lambda X_0.6X_0^2+5$. Consequently, $\sem{\M} \in \FP_1$.
\end{exa}

\section{Basic Feasible Functionals}

The class of tractable type 2 functionals has been introduced by Constable and Mehlhorn~\cite{Con73,Meh74}.
It was later named \bff\ for the class of Basic Feasible Functionals and characterized in terms of function algebra~\cite{CK89,IRK01}.
We choose to define the class through a characterization by Bounded Typed Loop Programs from~\cite{IKR02} which extends the original \bff\ to any order.

\subsection{Bounded Typed Loop Programs}
\begin{defi}[BTLP]
A Bounded Typed Loop Program (BTLP) is a non-recursive and well-formed procedure defined by the grammar of Figure~\ref{btlp}.
\end{defi}
The well-formedness assumption is given by the following constraints: Each procedure is supposed to be well-typed with respect to simple types over $\D$, the set of natural numbers. 
When needed, types are explicitly mentioned in variables' superscript. 
Each variable of a BTLP procedure is bound by either the procedure declaration parameter list, a local variable declaration or a lambda abstraction. In a loop statement, the guard variables $v_0$ and $v_1$ cannot be assigned to within $I^*$. In what follows $v_1$ will be called the \emph{loop bound}.

\begin{figure*}[t]
\hrule
\vspace*{0.3cm}
  \begin{align*}
\text{(Procedures)} &\ni P &&::= \textbf{Procedure } v^{\tau_1 \times \ldots \times \tau_n \to \D}( v_1^{\tau_1},\ldots,v_n^{\tau_n})\  P^*\ V\ I^* \textbf{ Return } v_r^\D \\
\text{(Declarations)} &\ni V &&::= \textbf{Var } v_1^\D, \ldots, v_n^\D;\\
\text{(Instructions)} &\ni I &&::= v^\D := E; \ | \ \textbf{Loop } v_0^\D \textbf{ with }v_1^\D \textbf{ do }\{I^* \} \\
\text{(Expressions)} &\ni E &&::=  0 \ | \ 1 \ |\ v^\D \ | \ v_0^\D + v_1^\D \ | \  v_0^\D - v_1^\D \ | \ v_0^\D \# v_1^\D \ | \ v^{\tau_1 \times \ldots \times \tau_n \to \D}(A_1^{\tau_1},\ldots,A_n^{\tau_n})\\
\text{(Arguments)} &\ni A &&::= v \ | \ \lambda v_1,\ldots,v_n.v(v'_1\ldots,v'_m) \quad \text{with }v \notin \{v_1,\ldots,v_n\}
\end{align*}
\hrule
\caption{BTLP grammar}
\label{btlp}
\end{figure*}

The operational semantics of BTLP procedures is standard: parameters are passed by call-by-value. $+$, $-$ and $\#$ denote addition, proper subtraction and smash function (\textsl{i.e.} $x \# y = 2^{\taille{x} \times \taille{y}}$, the size $\taille{x}$ of the number $x$ being the size of its dyadic representation), respectively. Each loop statement is evaluated by iterating $\taille{v_0}$-many times the loop body instruction under the following restriction: if an assignment $v:=E$ is to be executed within the loop body, we check if the value obtained by evaluating $E$ is of size smaller than the size of the loop bound $\taille{v_1}$. If not then the result of evaluating this assignment is to assign $0$ to $v$.

A BTLP procedure $\textbf{Procedure } v^{\tau_1 \times \ldots \times \tau_n \to \D}( v_1^{\tau_1},\ldots,v_n^{\tau_n})\  P^*\ V\ I^* \textbf{ Return } v_r^\D$ computes an order $i$ functional if and only if $\ord(\tau_1 \times \ldots \times \tau_n \to \D)=i$. The order of BTLP types can be defined similarly to the order on types of Definition~\ref{def:order}, extended by $\ord(\tau_1\times \ldots \times \tau_n)=\max_{i=1}^n(\ord(\tau_i))$ and where $\D$ is considered to be a basic type.

\begin{figure}
\hrule
\vspace*{0.3cm}
\begin{lstlisting}[morekeywords={Procedure,Return,Loop,with,do,Var,If,then}]
Procedure SumUp(F$^{(\D\to\D)\to\D}$,x$^\D$)
   Var bnd, maxi, sum, i;
   maxi := 0;
   i := 0;
   Loop x with x do {
      If F($\lambda$z.maxi) < F($\lambda$z.i) {
         maxi := i;
      }
      i := i+1;
   }
   bnd := (x+1)$\#($F($\lambda$z.maxi)+1);
   sum := 0;
   i := 0;
   Loop x with bnd do {
      sum := sum + F($\lambda$z.i);
      i := i+1;
   }
   Return sum$^\D$
\end{lstlisting}
\hrule
\caption{Example BTLP procedure}
\label{fig:sumup}
\end{figure}

An example BTLP procedure is provided in Figure~\ref{fig:sumup}. This example procedure {SumUp} is directly sourced from~\cite{IKR02} and is an order 3 procedure that computes the function: $
F, x \mapsto \sum_{i<|x|}F(\lambda z. i)$. It first computes a bound \verb:bnd: on the result by finding the number \verb:i: for which $F(\lambda z. i)$ is maximal and then computes the sum itself. Note that this procedure uses a conditional statement (\textbf{If}) not included in the grammar but that can be simulated in BTLP (see~\cite{IKR02}). Such a conditional will be explicitly added to the syntax of IBTLP procedures in Definition~\ref{defcheck}.

Let the size of an argument be the number of syntactic elements in it. The size of input arguments is the sum of the size of the arguments in the input.

\begin{defi}[Time complexity]
For a given procedure $P$ of parameters $( v_1^{\tau_1},\ldots,v_n^{\tau_n})$, we define its time complexity $t(P)$ to be a function of type $\bN \to \bN$ that, given an input of type $\tau_1 \times \ldots \times \tau_n$ returns the maximal number of assignments executed during the evaluation of the procedure in the size of the input.
\end{defi}

We are now ready to provide a definition of Basic Feasible Functionals at any order.

\newcommand{\bffi}{{\sc bff}$_i$}
\begin{defi}[\bffi]
For any $i\geq 1$, \bffi\ is the class of order $i$ functionals computable by a BTLP procedure\footnote{As demonstrated in~\cite{IKR02}, all types in the procedure can be restricted to be of order at most $i$ without any distinction.}.
\end{defi}
It was demonstrated in~\cite{IRK01} that \bff$_1$ =  {\sc FPtime} and \bff$_2$=\bff.

\subsection{Safe Feasible Functionals}
\newcommand{\sff}{{\sc {sff}}}
Now we restrict the domain of {\bff}$_i$ classes to inputs in \bff$_k$ for $k<i$, the obtained classes are named {\sff} for Safe Feasible Functionals.

\begin{defi}[{\sff}$_i$]
{\sff}$_1$ is defined to be the class of order $1$ functionals computable by BTLP a procedure and, for any $i\geq  1$, {\sff}$_{i+1}$ is the class of order $i+1$ functionals computable by BTLP a procedure on the input domain {\sff}$_i$. In other words,
\begin{align*}
\text{\sc sff}_1=&\text{\sc bff}_1,\\
\text{\sc sff}_{i+1}=&\text{\sc bff}_{{i+1}_{ \upharpoonright \text{\sc sff}_{i}}},\quad \forall i,\ i \geq 1.
\end{align*}
\end{defi}

This is not a huge restriction since we want an arbitrary term of a given complexity class at order $i$ to compute over terms that are already in classes of the same family at order $k$, for $k<i$. Consequently, programs can be built in a constructive way component by component. Another strong argument in favor of this domain restriction is that the partial evaluation of a functional at order $i$ will, at the end, provide a function in $\bN \longrightarrow \bN$ that is shown to be in {\bff}$_1$ (={\sc FPtime}). 

\section{A characterization of safe feasible functionals of any order}

In this section, we show our main characterization of safe feasible functionals:
\begin{thm}\label{main}
For any order $i \geq 0$, the class of functions in $\FP_{i+1}$ over $\FP_{k}$, $k \leq i$, is exactly the class of functions in {\sff}$_{i+1}$, up to an isomorphism.
In other words, $$\text{\sff}_{i+1} \equiv \FP_{i+1 \upharpoonright (\cup_{k \leq i} \FP_{k})},$$ for all $i \geq 0$, up to an isomorphism.
\end{thm}
\begin{proof}
For a fixed $i$, the theorem is proved in two steps: Soundness, Theorem~\ref{th:sound}, and Completeness, Theorem~\ref{th:comp}.
Soundness consists in showing that any term $\M$ whose interpretation is bounded by an order $i$ polynomial, computes a function in \sff$_i$. Completeness consists in showing that any BTLP procedure $P$ of order $i$ can be encoded by a term $\M$ computing the same function and admitting a polynomial interpretation of order $i$.
\end{proof}

Notice that functions in {\sff}$_{i+1}$ return the dyadic representation of a natural number. Consequently, the isomorphism is used on functions in $\FP_i$ to illustrate that a function of type $\bN \to (\bN \to \bN)$ and order $1$ is isomorphic to a functional of type $\bN \to \bN$ and of the same order using decurryfication and pair encoding over $\bN$. In order to simplify the treatment we will restrict ourselves to \emph{functional terms} computing functionals that are terms of type $\T \longrightarrow \texttt{b}$, with $\texttt{b} \in \B$, in the remaining of the paper.

\subsection{Soundness}

The soundness means that any term whose interpretation is bounded by an order $i$ polynomial belongs to \sff$_i$. For that, we will show that the interpretation allows us to bound the computing time with an higher order polynomial.

\begin{thm}\label{th:sound}
Any functional term $\M$ whose interpretation is bounded by an order $i$ polynomial, computes a functional in \sff$_i$.
\end{thm}

\begin{proof}
For order $1$, consider that the term $\M$ has an interpretation bounded by a polynomial $P_1$. For any value $\vone$, we have, by Corollary~\ref{coro1}, that the computing time of $\M$ on input $\vone$ bounded by $\interpone{\M\ \vone}$. Consequently, using Lemma~\ref{lem7}, we have that:
$$\interpone{\M\ \vone} = \interpone{\M} \interpone{\vone} = \interpone{\M}(|\vone|) \leq P_1(|\vone|).$$

Hence $\M$ belongs to {\sc FPtime} $= \text{\sff}_1$.

Now, for higher order, let $\M$ be an order $i+1$ term of interpretation $\interpone{\M}$. There exists an order $i+1$ polynomial $P_{i+1}$ such that $\interpone{\M} \leq P_{i+1}$. We know that on input $\N$, $\M$ normalizes in $O(\interpone{\M\ \N})$, by Corollary~\ref{coro2}. Since $\N$  computes a functional $\llbracket \N \rrbracket \in \text{\sff}_{i}$ there is a polynomial $P_i$ such that $\interpone{\N} \leq P_i$, by induction on the order $i$. Consequently, we obtain that the maximal number of reduction steps is bounded polynomially in the input size by:
$$\interpone{\M\ \N} = \interpone{\M}\interpone{\N} \leq P_{i+1} \circ P_i,$$
that is, by a polynomial $Q_{i+1}$ of order $i+1$ defined by $Q_{i+1}=P_{i+1} \circ P_i$.
\end{proof}
The above result holds for terms over arbitrary basic inductive types, by considering that each value on such a type encodes an integer value.

\subsection{Completeness}

To prove that all functions computable by a BTLP program of order $i$ can be defined as terms admitting a polynomial interpretation, we proceed in several steps:
\begin{enumerate}
\item We show that it is possible to encode each BTLP procedure $P$ into an intermediate procedure $\transform{P}{}$ of a language called IBTLP (See Figure~\ref{lbtlp}) for If-Then-Else Bounded Typed Loop Program such that $\transform{P}{}$ computes the same function as $P$ using the same number of assignments (\textsl{i.e.} with the same time complexity). 
\item We show that we can translate each IBTLP procedure $\transform{P}{}$ into a flat IBTLP procedure $\overline{\transform{P}{}}$, \textsl{i.e.} a procedure with no nested loops and no procedure calls. 
\item Then we transform the flat IBTLP procedure $\overline{\transform{P}{}}$ into a ``local'' and flat IBTLP procedure $[\overline{\transform{P}{}}]_\emptyset$ checking bounds locally in each assignment instead of checking it globally in each loop. For that purpose, we use a polynomial time computable operator of the IBTLP language called $\ckbd$. The time complexity is then asymptotically preserved. 
\item Finally, we compile the IBTLP procedure $[\overline{\transform{P}{}}]_\emptyset$ into a term of our language and we use completeness for first order function to show that there is a term computing the same function and admitting a polynomial interpretation. This latter transformation does not change the program behaviour in terms of computability and complexity, up to a O, but it makes the simulation by a functional program easier as each local assignment can be simulated independently of the context.
\end{enumerate}
The 3 first steps just consist in transforming a BTLP procedure into a IBTLP procedure in order to simplify the compilation procedure of the last step. 
These steps can be subsumed as follows:
\begin{center}
\begin{tikzpicture}[node distance=3cm,
fleche/.style={->},
box/.style={draw,rounded corners,text width=11mm},
carre/.style={draw,minimum size=1cm}]
\node(a){Program};
\node (b) [node distance=2cm, right of=a]{$P$};
\node (c) [node distance=2cm, right of=b] {$\transform{P}{}$};
\node (d) [node distance=2cm, right of=c] {$\overline{\transform{P}{}}$};
\node (e) [node distance=2cm, right of=d] {$[\overline{\transform{P}{}}]_\emptyset$};
\node (f) [node distance=3cm, right of=e] {$\pcmp([\overline{\transform{P}{}}]_\emptyset)$};

\node[node distance=1cm] (a0) [above of=a] {Language};
\node[node distance=1cm] (b0) [above of=b] {BTLP};
\node[node distance=1cm] (c0) [above of=c] {IBTLP};
\node[node distance=1cm] (d0) [above of=d] {IBTLP};
\node[node distance=1cm] (e0) [above of=e] {IBTLP};
\node[node distance=1cm] (f0) [above of=f]  {$\mathcal{T}$};

\draw[fleche] (b) to node[above] {$\transform{}{}$} (c);
\draw[fleche] (c) to node[above] {flat} (d);
\draw[fleche] (d) to node[above] {local} (e);
\draw[fleche] (e) to node[above] {compile} (f);

\end{tikzpicture}
\end{center}
For each step, we check that the complexity in terms of reduction steps is preserved and that the transformed program computes the same function.

\subsubsection{From BTLP to IBTLP}
\begin{defi}[IBTLP]\label{defcheck}
A If-Then-Else Bounded Typed Loop Program (IBTLP) is a non-recursive and well-formed procedure defined by the grammar of Figure~\ref{lbtlp}.
\end{defi}

\begin{figure*}[t]
\hrule
\vspace*{0.3cm}
  \begin{align*}
   IP &::=  v^{\tau_1 \times \ldots \times \tau_n \to \D}( v_1^{\tau_1},\ldots,v_n^{\tau_n})\{IP^* IV II^* \}\textbf{ ret } v_r^\D\\
 IV &::= \textbf{var } v_1^\D, \ldots, v_n^\D;\\
 II&::= v^\D := IE^\D; \ |\  \textbf{loop } v_0^\D \textbf{ \{}II^* \textbf{\}} \ |\ \textbf{if}\ v^\D\ \textbf{ \{}\ II^* \textbf{\}}\ [\textbf{else \{}\ II^*\textbf{\}}] \ | \  (II^*)_{v^{\D}}\\
 IE &::= 0 \ | \ { 1} \ | \  v^\D \ | \ v_0^\D + v_1^\D \ | \  v_0^\D - v_1^\D \ | \ v_0^\D \# v_1^\D \ | \ v^{\tau_1 \times \ldots \times \tau_n \to \D}(IA_1^{\tau_1},\ldots,IA_n^{\tau_n}) \ | \ v_0^\D \times v_1^\D \ | \ \cut(v^\D)   \\
 &\phantom{::=}\ | \ \ckbd(E,X)\\
 IA &::= v \ | \ \lambda v_1,\ldots,v_n.v(v'_1\ldots,v'_m) \quad \text{with }v \notin \{v_1,\ldots,v_n\}
\end{align*}
\hrule
\caption{IBTLP grammar}
\label{lbtlp}
\end{figure*}

The well-formedness assumption and variable bounds are the same as for BTLP. In a loop statement, the guard variable $v_0$ still cannot be assigned to within $II^*$.
A IBTLP procedure $IP$ has also a time complexity $t(IP)$ defined similarly to the one of BTLP procedures. 

The main distinctions between an IBTLP procedure and a BTLP procedure are the following:
\begin{itemize}
\item there are no loop bounds in IBTLP loops. Instead loop bounds are written as instruction annotations: a IBTPL loop $(\textbf{loop } v_0^\D \textbf{ \{}II^*\textbf{\}})_{v^{\D}}$ corresponds to a BTLP instruction of the shape $$\textbf{Loop } v_0^\D \textbf{ with }v^{\D} \textbf{ do }\{II^* \}.$$ 
\item \itlp\ includes a conditional statement $\textbf{if}\ v^\D\ \textbf{\{}\ II_1^*\ \textbf{\} else \{}\ II_2^*\textbf{\}}$ evaluated in a standard way: if variable $v$ is $0$  then it is evaluated to $II_2^*$. In all other cases, it is evaluated to $II_1^*$, the else branching being optional.
\item \itlp\ includes a basic operator $\times$ such that $x \times y = 2^{\taille{x}+\taille{y}}$.
\item \itlp\ includes a unary operator $\cut$ which removes the first character of a number (\textsl{i.e.} $\cut(0)$ = $0$, $\cut(2x+i)=x$ where $i \in \{0,1\}$). 
\item \itlp\ includes an operator $\ckbd$ computing the following function:
\begin{align*}
\ckbd(E,X) =\left\{
\begin{array}{cr}
\sem{E}, &\text{if } \taille{\sem{E}} \leq  \taille{x}, x \in X, \\
0, &\text{otherwise},
\end{array}
\right.
\end{align*}
where $\sem{E}$ is the dyadic number obtained after the evaluation of expression $E$ and $X$ is a finite set of variables.
\end{itemize}
Notice that $\ckbd$ is in {\sff}$_1$ provided that the input $E$ is computable in polynomial time and both $\times$ and $\verb!cut!$ are in {\sff}$_1$.
The semantics of an IBTLP procedure is also similar to the one of a BTLP procedure: during the execution of an assignment, the bound check is performed on instruction annotations instead of being performed on loop bounds. However IBTLP is strictly more expressive than BTLP from an extensional perspective:  a loop can be unbounded. This is the reason why only IBTLP procedures obtained by well-behaved transformation from BTLP procedures will be considered in the remainder of the paper.

Now we define a program transformation $\transform{.}{}$ from BTLP to IBTLP.
For each loop of a BTLP program, this transformation just consists in recording the variable appearing in the $\textbf{ with }$ argument of a contextual loop and putting it into an instruction annotation as follows:
$$\transform{\textbf{Loop } v_0^\D \textbf{ with }v_1^\D \textbf{ do }\{I^* \}}{} = (\textbf{loop } v_0^\D \textbf{ \{}\transform{I}{}^*\textbf{\}})_{v_1^\D}.$$
Any assignment is left unchanged and this transformation is propagated inductively on procedure instructions so that any inner loop is transformed. We denote by $\transform{P}{}$ the IBTLP procedure obtained from the BTLP procedure $P$.
Hence the following lemma straightforwardly holds:
\begin{lem}\label{8}
Given a procedure $P$, let $\sem{P}$ denote the function computed by $P$. For any BTLP procedure $P$, we have $\sem{P} = \sem{\transform{P}{}} \text{ and } t(P) = t(\transform{P}{})$.
\end{lem}
\begin{proof}
The transformation is semantics-preserving (the computed function is the same).
Any assignment in $P$ corresponds to exactly one assignment in $\transform{P}{}$ and the number of iterations of the loop instructions $\textbf{Loop } v \textbf{ with } \ldots$ and $\textbf{loop } v$ are both in $\taille{v}$.
\end{proof}

\subsubsection{From IBTLP to Flat IBTLP}

\begin{defi}[Flat IBTLP]\label{deffclbtlp}
A If-Then-Else Bounded Typed Loop Program (IBTLP) is \emph{flat} if it does not contain nested loops.

\end{defi}
We will show that it is possible to translate any IBTLP procedure into a Flat IBTLP procedure using the \verb!if! construct. 

There are only three patterns of transformation: 
\begin{enumerate}
\item one pattern for nested loops, called \emph{Unnest},
\item  one pattern for sequential loops, called \emph{Unseq},
\item and one for procedure calls inside a loop, called \emph{Unfold},
\end{enumerate}
that we describe below:
\begin{enumerate}
\item The first transformation \emph{Unnest} consists in removing a nested loop of a given procedure. Assume we have a IBTLP procedure with two nested loops:

\begin{lstlisting}
(loop x { $II_1^*$ (loop y {$II_2^*$})$_z$ $II_3^*$})$_\text{w}.$
\end{lstlisting}

We can translate it to a IBTLP procedure as:

\begin{lstlisting}
total := w$\times$y; dx := 1; gt := total; gy := y; 
lb := x$\#$total;
loop lb {
  if dx {($II_1^*$)$_\text{w}$}
  if gy {(($II_2^*$)$_\text{z}$)$_\text{w}$ gy := cut(gy);}
  if dx {dx := $0$; ($II_3^*$)$_\text{w}$}
  if gt {gt := cut(gt);}
  else {gt := total; gy := y; dx := 1;}
}
\end{lstlisting}
where all the variables \verb!dx!, \verb!total!, \verb!gt!, \verb!gy! and \verb!lb! are fresh local variables and $II_1^*$, $II_2^*$ and $II_3^*$ are instructions with no loop.
The intuition is that variables \verb!dx! and \verb!gy! tell whether the loop should execute $II_1^*$ and $II_2^*$ respectively. Variable \verb!gt! counts the number of times to execute the sub-loop.

\item The second transformation \emph{Unseq} consists in removing two sequential loops of a given procedure. Assume we have a IBTLP procedure with two sequential loops: 

\begin{lstlisting}
(loop x$_1$ {$II_1^*$})$_{\text{w}_1}$ $II_2^*$ (loop x$_3$ {$II_3^*$})$_{\text{w}_3}.$
\end{lstlisting}

We can translate it to a single loop IBTLP procedure as:

\begin{lstlisting}
gx := x1; dy := 1; lb := x$_1 \times$x$_3$; 
loop lb {
  if gx {gx := cut gx; ($II_1^*$)$_{\text{w}_1}$} 
  else  {
    if dy {$II_2^*$}
    else {dy := $0$; ($II_3^*$)$_{\text{w}_3}$}
  }
}
\end{lstlisting}
where \verb!gx!, \verb!dy! and \verb!lb! are fresh local variables and $II_1^*$, $II_2^*$ and $II_3^*$ are instructions with no loop.

\item The last transformation \emph{Unfold} consists in removing one procedure call of a given procedure. Assume we have a IBTLP procedure with a call to procedure $P$ of arity $n$: $X \verb! := ! P(IA_1,...,IA_n)$;

We can translate it to a computationally equivalent IBTLP procedure after removing the call to procedure $P$. For that purpose, we carefully alpha-rename all the variables of the procedure declaration (parameters and local variables) to obtain the procedure $P( v_1,\ldots,v_n) \{IP^* IV  II^*\}$ $\textbf{ ret } v_r$ then we add the procedure declarations $IP^*$ to the caller procedure list of procedure declarations, and the local variables $IV$ and parameters $v_1,\ldots,v_n$ to the caller procedure list of local variable declarations and then we generate the following code:

\begin{lstlisting}
$v_1$ := $IA_1$; $\ldots$ $v_n$ := $IA_n$; $II^*$ $X \verb! := ! v_r$;
\end{lstlisting}
This program transformation can be extended straightforwardly to the case where the procedure call is performed in a general expression. Notice that unfolding a procedure is necessary as nested loops may appear because of procedure calls.
\end{enumerate}

These three patterns can be iterated on a IBTLP procedure (from top to bottom) to obtain a Flat IBTLP procedure in the following way:
\begin{defi}
The transformation $\overline{IP}$ is a mapping from IBTLP to IBTLP defined by:
$$\overline{IP}=(\emph{Unnest}^! \circ \emph{Unseq}^!)^! \circ \emph{Unfold}^! (IP),$$
where, for a given function f, $f^!$ is the least fixpoint of $f$ on a given input and $\circ$ is the usual mapping composition. 
\end{defi}
Notice that this procedure is polynomial time computable as each application of a call to \emph{Unfold} consumes a procedure call (and procedures are non recursive) and each application of an \emph{Unnest} or \emph{Unseq} call consumes one loop. Consequently, fixpoints always exist.

\begin{lem}\label{fclbtlp1}
For any BTLP procedure $P$,  $\overline{\transform{P}{}}$ is a flat IBTLP procedure.
\end{lem}
\begin{proof}
First notice that repeated application of the \emph{Unfold} pattern may only introduce a constant number of new loops as procedures are non-recursive. Consequently, the fixpoint $\emph{Unfold}^!$ is defined and reached after a constant number (in the size of the procedure) of applications. Each application of a $\emph{Unseq}$ pattern or $\emph{Unnest}$ pattern decreases by one the number of loop within a procedure. Consequently, a fixpoint is reached (modulo $\alpha$-equivalence) and the only programs for which such patterns cannot apply are programs with a number of loops smaller than 1.
\end{proof}

\begin{lem}\label{fclbtlp2}
For any BTLP procedure $P$, we have $ \sem{\overline{\transform{P}{}}} = \sem{P}\text{ and } t(\overline{\transform{P}{}}) = O(t(P))$.

\end{lem}
\begin{proof}
In the first equality, the computed functions are the same since the program transformation preserves the extensional properties. For the second equality, the general case can be treated by a structural induction on the procedure $\transform{P}{}$. For simplicity, we consider the case of a procedure $\transform{P}{}$ only involving $n$ nested loops over guard variables $x_1,\ldots,x_n$ and loop bounds $x_{n+1},\ldots,x_{2n}$, respectively over one single basic instruction (\textsl{e.g.} one assignment with no procedure call).  With inputs of size $m$, this procedure will have a worst case complexity of $m^n$ (when the loop bounds are not reached). Consequently, $t(P)(m)=m^n$. This procedure will be transformed into a flat IBTLP procedure using the \emph{Unnest} transformation $n-1$ times over a variable $z$ containing the result of the computation $x_1 \# ((x_2\times x_{n+1}) \# ( \ldots x_3 \times x_{n+2})\# (x_n \times x_{2n} )\ldots)$ as initial value. Consequently, on an input of size $m$, we have $t(\overline{\transform{P}{}})(m) \leq \taille{z}= \taille{x_1} \times (\taille{x_2}+ \taille{x_{n+1}})\times \ldots \times (\taille{x_n}+\taille{x_{2n}}) \leq 2^{n-1} \times m^n =O(m^n)$, as for each $i$ $\taille{x_i} \leq m$ and by definition of the $\#$ and $\times$ operators.  
We conclude using Lemma~\ref{8}.
\end{proof}

\subsubsection{From Flat IBTLP to Local and Flat IBTLP}
Now we describe the last program transformation $[ - ]_X$ that makes the check bound performed on instruction annotations explicit. For that purpose, $[ - ]_X$ makes use of the operator $\ckbd$ and records a set of variables $X$ (the annotations enclosing the considered instruction). The procedure is described in Figure~\ref{local}.
\begin{figure*}[t]
\hrule
\vspace*{0.3cm}
  \begin{align*}
[   v^{\tau_1 \times \ldots \times \tau_n \to \D}( v_1^{\tau_1},\ldots,v_n^{\tau_n}) IP^*\ IV\ II^* \textbf{ ret } v_r^\D ]_X
&= v^{\tau_1 \times \ldots \times \tau_n \to \D}( v_1^{\tau_1},\ldots,v_n^{\tau_n}) IP^*\ IV\ [II]_X^* \textbf{ ret } v_r \\
[II_v]_X &= [II]_{X \cup\{v\}}\\
[\textbf{if}\ v^\D\ \textbf{\{}\ II_1^*\ \textbf{\} else \{}\ II_2^*\textbf{\}}]_X &= \textbf{if}\ v^\D\ \textbf{\{}\ [II_1]_X^*\ \textbf{\} else \{}\ [II_2]_X^* \textbf{\}}\\
[\textbf{loop } v_0^\D \textbf{\{}II^*\textbf{\}};]_X &= \textbf{loop } v_0^\D \textbf{ \{}[II]_X^* \textbf{\}}  \\
[v^\D := IE;]_X &= v^\D := \ckbd(IE^\D, X); 
\end{align*}
\hrule
\caption{From Flat IBTLP to Local and Flat IBTLP}
\label{local}
\end{figure*}
Notice that the semantics condition ensuring that no assignment must be performed on a computed value of size greater than the size of the loop bounds for a BTLP procedure or than the size of the loop annotations for a IBTLP procedure has been replaced by a local computation performing exactly the same check using the operator $\ckbd$. Consequently, we have:
\begin{lem}\label{fclbtlp3}
For any BTLP procedure $P$, we have $\sem{[\overline{\transform{P}{}}]_\emptyset} = \sem{P}$ and  $t([\overline{\transform{P}{}}]_\emptyset) = O(t(P))$.
\end{lem}
\begin{proof}
We use Lemma~\ref{fclbtlp2} together with the fact that extensionality and complexity are both preserved.
\end{proof}

\subsubsection{From Local and Flat IBTLP to terms}
\newcommand{\ourlang}{Esperanto}
We then encode Flat and Local IBTLP in our functional language. For this, we define a procedure $\comp$ that will ``compile'' IBTLP procedures into terms.
For that purpose, we suppose that the target term language includes constructors for numbers ($0$ and $1$), a constructor for tuples $\langle \ldots \rangle$, all the operators of IBTLP as basic prefix operators (+, -, $\#$, $\ldots$), min operators $\sqcap^\D_n$ computing the $\min$ of the sizes of $n$ dyadic numbers and a $\ckbd$ operator of arity 2 such that $\sem{\ckbd}(\M,\N)=\sem{\M}$ if $\taille{\sem{\M}} \leq \taille{\sem{\N}}$ (and $0$ otherwise). All these operators are extended to be total functions in the term language: they return $0$ on input terms for which they are undefined. Moreover, we also require that the Flat and Local IBTLP procedures given as input are alpha renamed so that all parameters and local variables of a given procedure have the same name and are indexed by natural numbers.
The compiling process is described in Figure~\ref{comp}, defining $\pcmp, \icmp^n, \ecmp\ $that respectively compile procedures, instructions and expressions. The $\icmp$ compiling process is indexed by the number of variables in the program $n$. 

For convenience, let $\lambda \langle v_1, \ldots, v_n \rangle.$ be a shorthand notation for
$\lambda s. {\tt case}\ s\ {\tt of}{\ \langle v_1, \ldots, v_n\rangle}\to$ and let $\pi_r^n$ be a shorthand notation for the r-th projection $\lambda \langle v_1, \ldots, v_n \rangle.v_r$.

\begin{figure*}[t]
\small
\hrule
\vspace*{0.3cm}
  \begin{align*}
 &\pcmp( v( v_1,\ldots,v_m)\ IP^*\ \textbf{var}\ v_{m+1},\ldots,v_n; \ II^*\textbf{ ret } v_r) = \lambda s.\pi_r^{n}( \icmp^{n}(II^*)\ s)\\
&\icmp^n(II_1 \ \ldots \ II_k) = \icmp^n(II_k)\ \ldots \ \icmp^n(II_1)\\
&\icmp^n(v_i := IE;) =  \lambda {\langle v_1, \ldots, v_n\rangle}.{\langle v_1,\ldots, v_{i-1}, \ecmp(IE), v_{i+1},\ldots,v_n \rangle} \\
&\icmp^n(\textbf{loop}\ v_i\ \textbf{\{}II^*\textbf{\}}) = \lambda \langle v_1,\ldots, v_n\rangle .(\left(\begin{array}{rcl}\letrec{f}{\lambda \tilde{t}.\lambda s .  \ocase\ \tilde{t}\ {\tt of}\ 0 &\to & s\\
j(t) &\to & f\ t\ (\icmp^n(II^*)\ s)\\
}
\end{array}
\right)v_i)\\
&\icmp^n(\textbf{if}\ v_i^\D\ \textbf{\{}\ II_1^*\ \textbf{\} else \{}\ II_2^*\textbf{\}}) = \lambda \langle v_1, \ldots, v_n \rangle. {\tt case}\ v_i\ {\tt of }\ 0 \to \icmp^n(II_2^*) \ | \ j(t) \to \icmp^n(II_1^*)\\
& \text{ with } j \in \{0,1\}\\
&\ecmp(c) = c, \quad c \in \{1,v^\D\}\\
&\ecmp(v_0^\D\ \op\ v_1^\D) = \op\ v_0\ v_1, \quad \op \in \{+,-,\#,\times\}\\
&\ecmp(\cut(v^\D)) = \cut \ v\\
&\ecmp(\ckbd(IE, \{v_{j_1},\ldots,v{j_r}\})) = \ckbd\ \ecmp(IE)\ (\sqcap^\D_r \ v_{j_1}\ \ldots\ v_{j_r}) \\
 &\ecmp(v(IA_1,\ldots,IA_n))=v \ \ecmp(IA_1)\ \ldots \ \ecmp(IA_n)\\
&\ecmp(\lambda v_1,\ldots,v_n.v(v'_1\ldots,v'_m))=\lambda v_1.\ldots \lambda v_n.(\ldots(v\ v'_1)\ \ldots\ v'_m)\ldots) 
\end{align*}
\hrule
\caption{Compiler from Local and Flat IBTLP to terms}
\label{comp}
\end{figure*}

The compilation procedure works as follows: any Local and Flat IBTLP procedure of the shape $v( v_1,\ldots,v_m)\ IP^*\ \textbf{var}\ v_{m+1},\ldots,v_n;\ II^*\ \textbf{ ret } v_r$ will be transformed into a term of type $\tau_1 \times \ldots \times \tau_n \to \tau_r$, provided that $\tau_i$ is the type of $v_i$ and that $\tau_1 \times \ldots \times \tau_n$ is the type for n-ary tuples of the shape $\langle v_1,\ldots,v_n \rangle$. Each instruction within a procedure of type $\tau_1 \times \ldots \times \tau_n \to \tau_r$ will have type $\tau_1 \times \ldots \times \tau_n \to \tau_1 \times \ldots \times \tau_n$. Consequently, two sequential instructions $II_1 \ II_2$ within a procedure of $n$ variables will be compiled into a term application of the shape $\icmp^n(II_2)\ \icmp^n(II_1)$ and instruction type is preserved by composition. Each assignment of the shape $v_i := IE;$ is compiled into a term that takes a tuple as input and returns the identity but on the i-th component. The i-th component is replaced by the term obtained after compilation of $IE$ on which a checkbound is performed. The min operator applied for this checkbound is $\sqcap^{\D}_n$ whenever $X$ is of cardinality $n$. The compilation process for expressions is quite standard: each construct is replaced by the corresponding construct in the target language. Notice that the compiling procedure is very simple for procedures as it only applies to Flat IBTLP procedures on which any procedure call has been removed by unfolding. 
The only difficulty to face is for loop compilation: we make use of a {\tt letRec} of type $\D \to \tau_1 \times \ldots \times \tau_n \to \tau_1 \times \ldots \times \tau_n$. The first argument is a counter and is fed with a copy of the loop counter $v_i$ so that the obtained term has the expected type $\tau_1 \times \ldots \times \tau_n \to \tau_1 \times \ldots \times \tau_n$.

Again, for a given term $\M$ of type ${\tau_1} \ldots \longrightarrow \ldots {\tau_n} \longrightarrow \tau$, we define its time complexity $t(\M)$ to be a function of type $\bN \to \bN$ that, for any inputs of type $\tau_i$ $\vec{\N}$ and size bounded by $m$ returns the maximal value of $\taille{\M\ \vec{\N} \Rightarrow \sem{\M\ \vec{\N}}}$. In other words, $t(\M)(m)=\max_{\taille{\vec{\N}} \leq m}\{\taille{\M\ \vec{\N} \Rightarrow \sem{\M\ \vec{\N}}}\}$.
\begin{lem}\label{fclbtlp4}
For any BTLP procedure $P$, we have $ \sem{\pcmp([\overline{\transform{P}{}}]_\emptyset)} = \sem{P}$ and \\
$ t(\pcmp([\overline{\transform{P}{}}]_\emptyset)) = O(t(P))$.
\end{lem}
\begin{proof}
The term obtained by the compilation process is designed to compute the same function as the one computed by the initial procedure. It remains to check that for a given instruction the number of reductions remains in $O(t(P))$. This is clearly the case for an assignment as it consists in performing one beta-reduction, one case reduction, and the evaluation of a fixed number of symbols within the expression to evaluate. For an iteration the complexity of the {\tt letRec} call is in $\taille{v_i}$ as for an IBTLP loop. Indeed each recursive call consists in removing one symbol. Finally, it remains to apply lemma~\ref{fclbtlp3}.
\end{proof}

\begin{exa}
Let us take a simple BTLP procedure:
\begin{lstlisting}
Procedure mult(x,y)
Var z,b; z := $0$; b := x#y;
Loop x with b do{z := z+y;}
Return z
\end{lstlisting}
It will be translated in Local and Flat IBTLP as:
\begin{lstlisting}
mult(x,y){
    var z,b; z := $0$; b := x#y;
    loop x {z := chkbd(z+y,{b});}
    ret z}
\end{lstlisting}
And be compiled modulo $\alpha$-equivalence in a term as:
$$
\lambda s .\pi_3^4\left(\begin{array}{l}

  \lambda \langle \texttt{x}, \texttt{y}, \texttt{z}, \texttt{b} \rangle .
    \left(\letrec{f}{
        \begin{array}{rcl}
\lambda \tilde{t}.\lambda \tilde{s}. \texttt{case } \tilde{t} \texttt{ of } 0 &\to& s\\
              j(t) &\to& f\ t\ (\M\ \tilde{s})
      )\end{array}}
   \right) \texttt{x} \\
 ( \lambda \langle \texttt{x}, \texttt{y}, \texttt{z}, \texttt{b} \rangle . \langle \texttt{x}, \texttt{y}, \texttt{z}, \#\ \texttt{x}\ \texttt{y} \rangle (\lambda \langle \texttt{x}, \texttt{y}, \texttt{z}, \texttt{b} \rangle . \langle \texttt{x}, \texttt{y}, 0, b\rangle \ s))    
\end{array}
\right),
$$
where $\M = \lambda \langle \texttt{x}, \texttt{y}, \texttt{z}, \texttt{b} \rangle .\langle \texttt{x}, \texttt{y}, \ckbd\ (+\ \texttt{z}\ \texttt{y})\ (\sqcap_1^\D\ \texttt{b}), \texttt{b} \rangle .$
\end{exa}

Now it remains to check that any term obtained during the compilation procedure has a polynomial interpretation.
In order to show this result, we first demonstrate some intermediate results:
\begin{lem}\label{si}
Given an assignment $\rho$, the operators $+$, $-$, $\#$, $\times$, $\cut$, $\sqcap^\D_n$ and $\ckbd$ admit
a polynomial sup-interpretation.
\end{lem}
\begin{proof}
We can check that the following are polynomial sup-interpretations:
\begin{align*}
\interpone{+} &= \Lambda X.\Lambda Y. (X+Y),\\
\interpone{-} &= \Lambda X.\Lambda Y. X,\\
\interpone{\#} &= \Lambda X.\Lambda Y. (X \times Y),\\
\interpone{\times} &= \Lambda X.\Lambda Y. (X + Y),\\
\interpone{\cut} &= \Lambda X.\sqcup(X-1,0),\\
\interpone{{\sqcap^\D_n}} &= \Lambda X_1.\ldots.\Lambda X_n. \sqcap^{\cN}(X_1,\ldots,X_n),\\
\interpone{{\ckbd}} &= \Lambda X.\Lambda Y. Y.
\end{align*}
The inequalities are straightforward for the basic operators $+,-,\#,\ldots$.
For $ \interpone{{\ckbd}}$, to be a sup-interpretation, we have to check that:
$$\forall \M,\forall \N,\, \interpone{{\ckbd}\ \M \ \N} \geq \interpone{\sem{\ckbd}(\M,\N)}.$$ 
Indeed,
\begin{align*}
 \interpone{{\ckbd}\ \M \  \N} &=\interpone{\ckbd}\ \interpone{\M}\ \interpone{\N}\\
 &=(\Lambda X.\Lambda Y. Y) \ \interpone{\M}\ \interpone{\N}\\
 &= \interpone{\N}\\
&\geq \interpone{\sem{\N}} &&\text{(By Corollary}~\ref{coro1}) \\
&\geq \taille{\sem{\N}} &&\text{(By Lemma}~\ref{lem7}) \\
&\geq \taille{\sem{\ckbd}(\M,\N)} &&\text{(By Definition of } \ckbd)\\
&\geq \interpone{\sem{\ckbd}(\M,\N)} &&\text{(By Lemma}~\ref{lem7})
\end{align*}
and so the conclusion. 

In the particular case where $\sem{\N}$ is not a value, we  have $\interpone{\sem{\ckbd}(\M,\N)}=\interpone{0}=1$ so the inequality is preserved.
\end{proof}

Now we are able to provide the interpretation of each term encoding an expression:
\begin{cor}\label{coro3}
Given a BTLP procedure $P$, any term $\M$ obtained by compiling an expression $IE$ (\textsl{i.e.} $\M=\ecmp( IE)$) of the flat and local IBTLP procedure $[\overline{\transform{P}{}}]_\emptyset$ admits a polynomial interpretation $\interpone{\M}$.
\end{cor}

\begin{proof}
By unfolding, there is no procedure call in the expressions of procedure $[\overline{\transform{P}{}}]_\emptyset$. Moreover, by Lemma~\ref{si} and by Definition~\ref{def:inter}, any symbol of an expression has a polynomial interpretation, obtained by finite composition of polynomials.
\end{proof}

\begin{thm}\label{th:comp}
Any BTLP procedure $P$ can be encoded by a term $\M$ computing the same function and admitting a polynomial interpretation.
\end{thm}
\begin{proof}
We have demonstrated that any loop can be encoded by a first order term whose runtime is polynomial in the input size. Each higher order expression in a tuple can be encoded by a first order term using defunctionalization. Consequently, by completeness of first-order polynomial interpretations there exists a term computing the same function and admitting a polynomial interpretation.
\end{proof}

\begin{exa}
We provide a last example to illustrate the expressive power of the presented methodology.  Define the term $\M$ by:
\begin{align*}
\letrec{\funone}{\lambda \funtwo.\lambda \x.\ccase{\x}\  {}\conone(\y,\z) &\to \funtwo\ (\funone\ \funtwo\ \z) },\\
 \nil &\to \nil.
\end{align*}
The interpretation of $\M$ has to satisfy the following inequality:
$$\sqcap \{ F \ | \ F \geq \Lambda G.\Lambda X. 4 \oplus  (1 \sqcup (\sqcup_{X \geq 1}(G( F \ G \ (X-1))))\}.$$
Clearly, this inequality does not admit any polynomial interpretation as it is at least exponential in $X$. Now consider the term $\M\ (\lambda \x.1+(1+(\x))$. The term $\lambda \x.1+(1+(\x))$ can be given the interpretation $\Lambda X. X\oplus 3$. We have to find a function $F$ such that $F\ (\Lambda X. X\oplus 3) \geq \Lambda Y. 4  \oplus (1 \sqcup (\sqcup_{Y \geq 1}(F\ (\Lambda X. X\oplus 3)  \ (Y-1))\oplus 3))$. This inequality is satisfied by the function $F$ such that $F\  (\Lambda X.X\oplus 3)\ Y= (7\times Y) \oplus 4$ and consequently $\M\ (\lambda \x.1+(1+(\x))$ has an interpretation. This highlights the fact that a term may have an interpretation even if some of its subterms might not have any. As expected, any term admitting an interpretation of the shape $\Lambda X. X\oplus \beta$, for some constant $\beta$, will have a polynomial interpretation when applied as first operand of this fold function.
\end{exa}

\section{Conclusion and future work}
This paper has introduced a theory for higher-order interpretations that can be used to deal with higher-order complexity classes. We manage to provide a characterization of the complexity classes \bffi\ introduced in~\cite{IKR02} for any order $i$, with a restriction on their input. For $i=1$, we obtain a characterization of \FPtime and, for $i=2$, we obtain a characterization of \bff\  applied to \FPtime inputs.

This is a novel approach but there are still some important issues to discuss.
\begin{itemize}
\item Synthesis: it is well-known for a long time that the synthesis problem that consists in finding the sup-interpretation of a given term is undecidable in general for first order terms using polynomial interpretations over natural numbers (see~\cite{P13} for a survey). As a consequence this problem is also undecidable for higher order. Some simplification such as defunctionalization and the use of interpretations over the reals can vanish this undecidability trouble.
\item Intensionality: the expressive power of interpretations in terms of captured algorithms (also called intensionality) is as usual their main drawback. As for first order interpretations, a lot of interesting terms computing polynomial time functions will not have any polynomial in\-ter\-pre\-tation, \textsl{i.e.} their interpretation will sometimes be $\top$, although the function will be computed by another algorithm (or term) admitting a finite interpretation. At least, the presented paper has shown that the expressive power of interpretations can be extended to higher order and we have presented some relaxation procedure to infer the interpretation of a term.
\end{itemize}

We now discuss some possible future works.
\begin{itemize}
\item It would be of interest to develop tractable (decidable) techniques to characterize the complexity classes \bffi, at least for $i=2$. Due to the intractability of the synthesis mentioned above, we have no expectation to solve this problem using interpretation methods. However other implicit complexity techniques such as tiering or light logics are candidates for solving this issue.
\item Space issues were not discussed in this paper as there is no complexity theory for higher order polynomial space. However one could be interested in certifying program space properties. In analogy with the usual first order theory, a suitable option could be to consider (possibly non-terminating) terms admitting a non strict polynomial interpretation. By non strict, we mean, for example, that the last rule of Definition~\ref{def:inter}  can be replaced by:
$$ \interpcontext{\letrec{\funone}{\M}}{\rho}=1 \oplus \sqcap\{ F \in {\interp{\T}}  \ | \ F \geq ((\Lambda \interpone{\funone}. \interpone{\M}) \ F)\}.$$
This would correspond to the notion of quasi-interpretation on first order programs~\cite{BMM11}. Termination is lost as the term $\letrec{\funone}{\funone}$ could be interpreted by $1 \oplus \Lambda F.F$. However a result equivalent to Lemma~\ref{lem2} holds: we still keep an upper bound on the interpretation of any derived term (hence on the ``size'' of such a term).
\end{itemize}

\end{document}